\documentclass [reqno,11pt]{amsart}
\usepackage{amsmath}%
\usepackage{amsfonts}%
\usepackage{amssymb}%
\usepackage{graphicx,epstopdf}
\usepackage[all]{xy}
\newtheorem{theorem}{Theorem}
\theoremstyle{plain}

\newtheorem{definition}{Definition}

\newtheorem{proposition}{Proposition}

\numberwithin{equation}{section}
\newcommand{\g}{\mathcal{G}}
\begin{document}
\title[]{ On the Contraction of $so(4)$ to $iso(3)$}
\author{Eyal M. Subag}
\address{Department of Mathematics, Technion-Israel Institute of \newline
\indent  Technology, Haifa 32000, Israel }
\email{esubag@tx.technion.ac.il}%
%\urladdr{http://www.authorone.uni-aone.de}
%\author{Author Two}
%\curraddr[A. Mann]{Author Two current address, line 1\newline%
%\indent Author Two current address, line 2}%
%\email[A.~Two]{author-two@authortwo-inst.hu}%
%\urladdr{http://www.authortwo.uni-atwo.hu}
\author{Ehud Moshe Baruch}
\address{Department of Mathematics, Technion-Israel Institute of \newline%
\indent Technology, Haifa 32000, Israel }
\email{embaruch@math.technion.ac.il}%
\author{Joseph L. Birman}
\address{Department of Physics,
The City College of the City University of  \newline%
\indent New York, New York N.Y.10031, USA \symbolfootnote[2]{Permanent Address} \newline
\indent and Department of Physics, Technion-Israel Institute of \newline%
\indent Technology, Haifa 32000, Israel \symbolfootnote[1]{Visiting Professor} }
\email {birman@sci.ccny.cuny.edu}%
\author{Ady Mann}
\address
{Department of Physics, Technion-Israel Institute of \newline%
\indent Technology, Haifa 32000, Israel}
\email{ady@physics.technion.ac.il}%
%\urladdr{http://www.authorthree.uni-athree.edu}
%\urladdr{http://www.authorfour.uni-athree.edu}
%\thanks{Thanks for Author One.}
%\thanks{Thanks for Author Two.}
%\thanks{This paper is in final form and no version of it will be submitted for
%publication elsewhere.}
%\date{March 15, 2002}
%\subjclass{Primary 05C38, 15A15; Secondary 05A15, 15A18} %
%\keywords{Keyword one, keyword two etc.}%
%\dedicatory{Dedicated to Professor XY on the occasion of his seventieth birthday.}
\maketitle
\begin{abstract}
For any skew-Hermitian integrable irreducible infinite dimensional representation $\eta$ of $iso(3)$, we find a sequence of (finite dimensional) irreducible representations $\rho_n$ of $so(4)$ which contract to $\eta$.
\end{abstract}
%\pacs{02.20.Qs, 02.20.Sv}
%\maketitle

\section{Introduction}

One of the first known  examples of contraction of Lie algebra representations, given in the early work of \.{I}n\"{o}n\"{u} and Wigner \cite{IW1}, is the contraction of the representations of the Lie algebra $so(3)$ to those of $iso(2)$. In that example, starting from a sequence, $\left\{\rho_j\right\}_{j=1}^{\infty}$ of finite dimensional representations of $so(3)$ with increasing dimension,    they obtained an infinite dimensional representation, $\eta_q$ of $iso(2)$. They proved the contraction of the representations  by the following type of convergence of matrix elements:
\begin{eqnarray}
\lim_{j\longrightarrow \infty}	\left\langle ^j_{m'}| \rho_j(X)|^j_m\right\rangle=\left\langle ^q_{m'}| \eta_q(X)|^q_m\right\rangle
	\label{eq:1}
\end{eqnarray}
where $|^j_m\rangle$ (respectively  $|^q_m\rangle$)  is an element in an orthonormal basis of $\rho_j$ (respectively $\eta_q$).

In this paper we show that the same type of convergence of matrix elements as in (\ref{eq:1}), holds for the contraction of the finite dimensional irreducible representations of $so(4)$ to infinite dimensional irreducible representations of $iso(3)$.
The convergence is proved  using a less familiar description of the irreducible representations of $so(4)$ and $iso(3)$, due to  Pauli \cite{Pau}.
%It seems impossible to carry out the contraction in the  more common description of the representations of $so(4)$ as tensor products of representations of $so(3)$.
\long\def\symbolfootnote[#1]#2{\begingroup%
\def\thefootnote{\fnsymbol{footnote}}\footnote[#1]{#2}\endgroup}

Our paper is divided as follows: In sections 2 and 3 we will describe the representation theory of $so(4)$ and $iso(3)$ respectively. In section 4 we give the contraction of the algebra $so(4)$ to $iso(3)$ and prove the convergence of the appropriate matrix elements.
\setcounter{footnote}{0}
\section{Representation theory  of  $so(4)$.}
In this section and the one that follows we describe all the skew-Hermitian irreducible finite dimensional  representations of $so(4)$ and all the skew-Hermitian irreducible infinite dimensional representations of $iso(3)$. We recall that the Lie algebra $so(4)$ is the direct sum of two copies of the Lie algebra $so(3)$. Moreover every irreducible representation of $so(4)$ is a tensor product of two irreducible representations of $so(3)$. From the work of Weimar-Woods \cite{Wei1,Wei2} we know all the contractions of representations of $so(3)$ and hence we also know all the contractions of representations of $so(4)$ that respect the decomposition $so(4)=so(3)\oplus so(3)$. As noted in \cite{Wei1}, the contraction of $so(4)$ to $iso(3)$ does not respect this decomposition. Hence we use another description of the representations of $so(4)$ which was given by Pauli \cite{Pau}. The resemblance of the representations of $so(4)$ and $iso(3)$ in this description is more convenient for the contraction procedure.  We also give the relation between the parameterization of the irreducible representations as was given by Pauli \cite{Pau} and the more usual parameterization as a tensor product of two irreducible representations of $so(3)$.  \newline

The Lie algebra so(4) can be defined  by the basis $\left\{M_{1} ,M_{2} ,M_{3} ,N_{1} ,N_{2} ,N_{3} \right\}$  satisfying the following commutation relations:
\begin{eqnarray}
&&[M_i,M_j]=i\epsilon_{ijk}M_k \\
&& [N_i,N_j]=i\epsilon_{ijk}M_k  \\
&&[M_i,N_j]=i\epsilon_{ijk}N_k
\end{eqnarray}
where $\epsilon_{ijk}$ is the Levi-Civita totally antisymmetric symbol.
We will describe all the irreducible finite dimensional integrable representations of $so(4)$ in terms of another  basis which is $\left\{M_+ ,M_- ,M_{3} ,N_+ ,N_- ,N_{3} \right\}$, where $M_\pm=M_{1}\pm iM_{2} ,N_\pm=N_{1}\pm iN_{2}$.
$so(4)$ has two independent invariants (Casimir operators) $\vec{M} \cdot \vec{N}$ and $\frac{1}{2}(M^2 + N^2 )$\footnote{$\vec{M} \cdot \vec{N}$ and $M^2 + N^2$ are elements of the center of the universal enveloping algebra of $so(4)$ and  are given by: $\vec{M}=\left(M_1,M_2,M_3\right)$, $ \vec{N}=\left(N_1,N_2,N_3\right)$, $\vec{M} \cdot \vec{N}=\sum_{i=1}^3M_iN_i$, $M^2 + N^2=\sum_{i=1}^3\left(M^2_i+N^2_i\right)$}. On each irreducible representation of $so(4)$, $\vec{M} \cdot \vec{N}$ and $M^2 + N^2$ act as scalar operators with scalars which we denote by $G$ and $F$ respectively.
These two scalars determine uniquely (up to an isomorphism) the irreducible representation of $so(4)$.
We denote the irreducible representation of $so(4)$ with $F$ and $G$ by $\rho_{F,G}:so(4)\longrightarrow gl(V_{F,G})$. The representation space $V_{F,G}$ has an orthonormal basis of the form $\left\{|^{F,G}_{j,m_j} \rangle:j\in \left\{j_0, j_0+1,...,n\right\}, m_j\in \left\{-j,-j+1,...,j\right\}\right\}$, where  $j_0,n\in \left\{0,\frac{1}{2},1,\frac{3}{2},2,...\right\}$ and they satisfy $G^2=j_0^2 (n+1)^2$, $2F=j_0^2 +(n+1)^2 -1$.
The dimension of $V_{F,G}$ is given by $\sum_{i=j_0}^{n}(2i+1)= (n+1)^2-j_0^2 $.  The representation $\rho_{F,G}$ is given by:
\begin{eqnarray}
&& \rho_{F,G}(M_{3})|^{F,G}_{j,m_j} \rangle=m_j|^{F,G}_{j,m_j} \rangle \\
&& \rho_{F,G}(M_{\pm})|^{F,G}_{j,m_j} \rangle=\sqrt{(j\mp m_j)(j\pm m_j+1)}|^{F,G}_{j,m_j\pm 1} \rangle
\end{eqnarray}
\begin{eqnarray}
\rho_{F,G} (N_{3})|^{F,G}_{j,m_j} \rangle=&&\alpha^{F,G}_{j}\sqrt{(j+m_j)(j-m_j)}|^{F,G}_{j-1,m_j} \rangle\\ \nonumber
&& +\beta^{F,G}_j m_j|^{F,G}_{j,m_j} \rangle+\\  \nonumber
&&\alpha^{F,G}_{j+1}\sqrt{(j+m_j+1)(j-m_j+1)}|^{F,G}_{j+1,m_j} \rangle
\end{eqnarray}
\begin{eqnarray}
\rho_{F,G}(N_{\pm})|^{F,G}_{j,m_j}\rangle=&& \pm \alpha^{F,G}_{j}\sqrt{(j\mp m_j)(j\mp m_j-1)}|^{F,G}_{j-1,m_j \pm 1}\rangle\\  \nonumber
&& +\beta^{F,G}_j\sqrt{(j \mp m_j)(j\pm m_j+1)}|^{F,G}_{j,m_j \pm 1}\rangle\\  \nonumber
&& \mp \alpha^{F,G}_{j+1}\sqrt{(j \pm m_j+1)(j \pm m_j+2)}|^{F,G}_{j+1,m_j  \pm 1}\rangle
\end{eqnarray}
where
\begin{equation}
\beta_j^{F,G}=\frac{G}{j(j+1)}
\end{equation}
\begin{equation}
\alpha_j^{F,G}=\sqrt{\frac{2F+1-j^2-\frac{G^2}{j^2})}{(2j+1)(2j-1)}}
\end{equation}
\begin{center}
\item \subsection{\textbf{$so(4)$ as the direct sum $so(3) \oplus so(3)$}}
\end{center}
We define $K_{i}\equiv \frac{1}{2}(M_{i}+N_{i}) ,L_{i}\equiv \frac{1}{2}(M_{i}-N_{i}) ,i=1,2,3$ and we get a new basis for $so(4)$,\newline $\left\{K_1,K_2,K_3,L_1,L_2,L_3\right\}$, satisfying the following commutation relations:
\begin{eqnarray}
&& \left[K_i,K_j\right]=i\epsilon_{ijk}K_k \\
&& \left[L_i,L_j\right]=i\epsilon_{ijk}L_k \\
&& \left[K_i,L_j\right]=0.
\end{eqnarray}
We see that either $\left\{K_1,K_2,K_3\right\}$ or $\left\{L_1,L_2,L_3\right\}$ span an ideal of $so(4)$, which is isomorphic to $so(3)$ and hence, $so(4)=so(3)\oplus so(3)$.
The invariant operators in terms of this basis are:
\begin{eqnarray}
&& \vec{M} \cdot \vec{N} =K^2-L^2 \\
&& \frac{1}{2}(M^2 + N^2 )=K^2+L^2
\end{eqnarray}
It is well known\footnote{For example \cite{Sin}.} that each irreducible finite dimensional  representation of $so(4)$ is a tensor product of two irreducible finite dimensional representations of $so(3)$. So for each irreducible representation $\rho_{F,G}:so(4)\longrightarrow gl(V_{F,G})$ there are some $k,l\in \left\{0, \frac{1}{2}, 1, \frac{3}{2} , 2,...\right\}$ such that the representation
$\rho^k\otimes \rho^l: so(3)\oplus so(3)\longrightarrow gl(V_k\otimes V_l)$ is isomorphic to $\rho_{F,G}$. The representation $(\rho^j,V_j)$ is the unique\footnote{There is only one for each positive integer dimension, up to an isomorphism of representations and these are all the finite dimensional irreducible representations of $so(3)$. See for example \cite{Bri}.} irreducible representation of $so(3)$ with dimension $2j+1$.\newline
For the irreducible representation of $so(4)$ from Pauli's description, $\rho_{F,G}$, which is isomorphic to  $\rho^k\otimes \rho^l$, we have the following relations:
\begin{eqnarray}
&& G=k(k+1)-l(l+1)=\pm j_0(n+1)  \\
&& F=k(k+1)+l(l+1)=\frac{j_0^2+(n+1)^2-1}{2}  \\
&& dimV_{F,G}=dim\left(V_K\otimes V_l\right)=(2k+1)(2l+1)=\\ \nonumber
&&(n+1)^2-j_0^2=\left(2\sqrt{1+2(F+G)}-1\right)\left(2\sqrt{1+2(F-G)}-1\right)\\
&& n=k+l=-2+\sqrt{1+2(F+G)}+\sqrt{1+2(	F-G)}  \\
&& j_0=|k-l|=|\sqrt{1+2(F+G)}-\sqrt{1+2(F-G)}|  \\
&& k=-1+\sqrt{1+2(F+G)}  \\
&& l=-1+\sqrt{1+2(	F-G)} % \\
%&& G>0\Longrightarrow l<k, k=\frac{n+j_0}{2}, l=\frac{n-j_0}{2}  \\
%&& G<0\Longrightarrow l>k, k=\frac{n-j_0}{2}, l=\frac{n+j_0}{2}  \\
%&& G=0\Longrightarrow l=k=\frac{n}{2}, j_0=0
\end{eqnarray}
\begin{eqnarray}
&& G>0\Longrightarrow l<k, k=\frac{n+j_0}{2}, l=\frac{n-j_0}{2}  \\
&& G<0\Longrightarrow l>k, k=\frac{n-j_0}{2}, l=\frac{n+j_0}{2}  \\
&& G=0\Longrightarrow l=k=\frac{n}{2}, j_0=0
\end{eqnarray}
The two pairs of parameters $(k,l)$ and $(F,G)$ are equivalent and knowing the value of one of these pairs determines uniquely the irreducible representation. The pair $(j_0,n)$ does not determine uniquely the irreducible representation, but the values of $(j_0,n)$ along with the knowledge of the sign of $G$ does.
\section{Representation theory  of $iso(3)$}
%Note that in this section the numeration of the equations is the same as in the previous section, in order to emphasize the %resemblance.\newline
The Lie algebra $iso(3)$ can be defined by the basis $\left\{J_{1} ,J_{2} ,J_{3} ,P_{1} ,P_{2} ,P_{3} \right\}$   satisfying the following commutation relations:
\begin{eqnarray}
&&[J_i,J_j]=i\epsilon_{ijk}J_k \\
&&[P_i,P_j]=0 \\
&&[J_i,P_j]=i\epsilon_{ijk}P_k
\end{eqnarray}
We will describe all the skew-hermitian irreducible integrable infinite dimensional  representations of $iso(3)$ in the basis  $\left\{J_+ ,J_- ,J_{3} ,P_+ ,P_- ,P_{3} \right\}$ where  $J_\pm=J_{1}\pm iJ_{2}$  , $P_\pm=P_{1}\pm iP_{2}$.
$iso(3)$ has two independent invariants (Casimir operators) $P^2$ and $\vec{J} \cdot \vec{P}$.
On each irreducible representation of $iso(3)$, $P^2$ and $\vec{J} \cdot \vec{P}$ act as scalar operators with the scalars which we denote by $p^2$ and $C$ respectively.
These two scalars determine uniquely (up to an isomorphism) the irreducible representation of $iso(3)$.
We denote the irreducible representation of $iso(3)$ with given  $p^2$ and $C$ by $\eta_{p^2,C}:iso(3)\longrightarrow gl(W_{p^2,C})$. The representation space  $W_{p^2,C}$ has an orthonormal basis of the form $\left\{|^{p^2,C}_{j,m_j} \rangle:j\in \left\{j_0, j_0+1,...\right\}, m_j\in \left\{-j,-j+1,...,j\right\}\right\}$, where $j_0\in \left\{0,\frac{1}{2},1,\frac{3}{2},2,...\right\}$ and they satisfy $C^2=j_0^2p^2$. All the $W_{p^2,C}$ are infinite dimensional.
The representation $\eta_{p^2,C}$ is given by:
\begin{eqnarray}
&& \eta_{p^2,C}(J_{3})|^{p^2,C}_{j,m_j} \rangle=m_j|^{p^2,C}_{j,m_j} \rangle \\
&& \eta_{p^2,C}(J_{\pm})|^{p^2,C}_{j,m_j} \rangle=\sqrt{(j\mp m_j)(j\pm m_j+1)}|^{p^2,C}_{j,m_j\pm 1} \rangle
\end{eqnarray}
\begin{eqnarray}
\eta_{p^2,C} (P_{3})|^{p^2,C}_{j,m_j} \rangle=&&\tilde{\alpha}^{p^2,C}_{j}\sqrt{(j+m_j)(j-m_j)}|^{p^2,C}_{j-1,m_j} \rangle\\ \nonumber
&& +\tilde{\beta}^{p^2,C}_j m_j|^{p^2,C}_{j,m_j} \rangle+\\  \nonumber
&&\tilde{\alpha}^{p^2,C}_{j+1}\sqrt{(j+m_j+1)(j-m_j+1)}|^{p^2,C}_{j+1,m_j} \rangle
\end{eqnarray}
\begin{eqnarray}
 \eta_{p^2,C}(P_{\pm})|^{p^2,C}_{j,m_j}\rangle=&& \pm \tilde{\alpha}^{p^2,C}_{j}\sqrt{(j\mp m_j)(j\mp m_j-1)}|^{p^2,C}_{j-1,m_j \pm 1}\rangle\\  \nonumber
&& +\tilde{\beta}^{p^2,C}_j\sqrt{(j \mp m_j)(j\pm m_j+1)}|^{p^2,C}_{j,m_j \pm 1}\rangle\\  \nonumber
&& \mp \tilde{\alpha}^{p^2,C}_{j+1}\sqrt{(j \pm m_j+1)(j \pm m_j+2)}|^{p^2,C}_{j+1,m_j  \pm 1}\rangle
\end{eqnarray}
where
\begin{equation}
\tilde{\beta}_j^{p^2,C}=\frac{C}{j(j+1)}
\end{equation}
\begin{equation}
\tilde{\alpha}_j^{p^2,C}=\sqrt{\frac{p^2-\frac{C^2}{j^2}}{(2j+1)(2j-1)}}
\end{equation}

\section{Contraction of the matrix elements}

In this section, we first  recall the definition for contraction and give the contraction of the algebra $so(4)$ to $iso(3)$. Then, for each of the representations $\eta_{p^2,C}$ we specify a suitable sequence of the representations $\rho_{F(n),G(n)}$ such that we obtain the desired convergence of matrix elements. We will not address the question of contraction of the group representations which was solved by  Dooley and Rice \cite{Doo} and was considered by others \cite{Wolf,Wong1,Wong2}.
 We note that a contraction of the representations of $so(3,1)$  to those of $iso(3)$ was done by Weimar-Woods \cite{Wei4}.  \newline

\begin{center}
\item \subsection{Contraction of $so(4)$ to $iso(3)$}
\end{center}
We recall the formal definition for a contraction of Lie algebras. Our notations are similar to those of Weimar-Woods ~\cite{Wei3}.
\begin{definition}
Let $U$ be a complex or real vector space. Let $\g=(U,[\_,\_])$ be a Lie algebra with Lie product $[\_,\_]$. For any $\epsilon \in (0,1]$ let $t_{\epsilon} \in Aut(U)$ ($t_{\epsilon}$ is a linear invertible operator on $U$) and for every $X,Y \in U$ we define
\begin{eqnarray}
[X,Y]_{\epsilon}=t^{-1}_{\epsilon}([t_{\epsilon}(X),t_{\epsilon}(Y)]).
\label{eq:11}
\end{eqnarray}
If the limit
\begin{eqnarray}
[X,Y]_0=\lim_{\epsilon \longrightarrow 0^{+}}[X,Y]_{\epsilon}
\label{eq4.2}
\end{eqnarray}
exists for all $X,Y \in U$, then $[\_,\_]_0$ is a Lie product on $U$ and the Lie algebra $\mathcal{G}_0=(U,[\_,\_]_0)$ is called the contraction of $\mathcal{G}$ by $t_{\epsilon}$ and we write $\mathcal{G}\stackrel{t(\epsilon)}{\rightarrow} \mathcal{G}_0$.
\end{definition}
There is an analogous definition ~\cite{Wei3} for the case that  the limit (\ref{eq4.2}) is meaningful only on a sequence:
\begin{definition}
Let $U$ be a complex or real  vector space, $\g=(U,[\_,\_])$ a Lie algebra with Lie product $[\_,\_]$. For any $n \in \mathbb{N}$ let $t_n \in Aut(U)$  and for every $X,Y \in U$ we define
\begin{eqnarray}
[X,Y]_{n}=t^{-1}_{n}([t_{n}(X),t_{n}(Y)]).
\label{eq:11}
\end{eqnarray}
If the limit
\begin{eqnarray}
[X,Y]_{\infty}=\lim_{n \longrightarrow \infty}[X,Y]_{n}
\label{eq:2}
\end{eqnarray}
exists for all $X,Y \in U$, then $[\_,\_]_{\infty}$ is a Lie product on $U$ and the Lie algebra $\mathcal{G}_{\infty}=(U,[\_,\_]_{\infty})$ is called the contraction of $\mathcal{G}$ by $t_{n}$ and we write $\mathcal{G}\stackrel{t_n}{\rightarrow} \mathcal{G}_{\infty}$
\end{definition}
Specific examples of contractions of Lie algebras can be found in e.g., \cite{IW1, Wei1, Sal, Gil}.\newline

For the $so(4)\rightarrow iso(3)$ case we define the contraction transformation to be $t_{\epsilon}(M_i)=M_i$, $t_{\epsilon}(N_i)=\epsilon N_i$ for every $i \in \left\{1,2,3 \right\} $.  Then we easily see that:
\begin{eqnarray}
&&[M_i,M_j]_0=i\epsilon_{ijk}M_k \\
&&[N_i,N_j]_0=0 \\
&&[M_i,N_j]_0=i\epsilon_{ijk}N_k
\end{eqnarray}
We recall that
\begin{eqnarray}
&&[J_i,J_j]=i\epsilon_{ijk}J_k \\
&&[P_i,P_j]=0 \\
&&[J_i,P_j]=i\epsilon_{ijk}P_k
\end{eqnarray}
and we see that the linear map $\psi$, from the contracted Lie algebra, $so(4)_0$  to $iso(3)$ which is defined by $\psi(M_i)=J_i$, $\psi(N_i)=P_i$  for $i \in \left\{1, 2, 3 \right\}$ is a Lie algebra isomorphism.
\begin{center}
\item \subsection{convergence of the matrix elements}
\end{center}
Fix a representation $\eta_{p_1^2,C_1}$ of $iso(3)$ and define
\begin{eqnarray}
j_0^1=\sqrt{\frac{C_1^2}{p_1^2}}
\label{eq4.11}
\end{eqnarray}
We define a sequence of representations  which consists of some of the representations  $\rho_{(F,G)}$, as follows. We take those $\rho_{(F,G)}$ such that the value of their  $j_0$ parameter equals $j_0^1$ and such that $sgn(G)=sgn(C_1)$. There is exactly one irreducible representation for each admissible value of $n$,
where the admissible values of $n$ are $I=\left\{j_0^1, j_0^1+1, j_0^1+2, ... \right\}$. We can describe this sequence by $\left\{(\rho_{(F(n),G(n)}, V_{F(n),G(n)}) \right\}_{n \in I}$ where
\begin{eqnarray}
&& G(n)=sgn{(C_1)} j^1_0(n+1) \\
&& F(n)=\frac{(j^1_0)^2+(n+1)^2-1}{2}
\end{eqnarray}

Before we prove the convergence of matrix elements we need the following technical proposition:
\begin{proposition}
For
\begin{eqnarray}
{\epsilon}_n=\sqrt{\frac{p_1^2}{2F(n)}}=\sqrt{\frac{p_1^2}{j_{0,1}^2+(n+1)^2-1}}
\label{Eq.75}
\end{eqnarray}
the following hold
\begin{eqnarray}
	&& \lim_{n\longrightarrow \infty}\epsilon_n \beta^{F(n),G(n)}_j=\tilde{\beta}^{p^2_1,C_1}_j \label{Eq4.15} \\
	&& \lim_{n\longrightarrow \infty}\epsilon_n \alpha^{F(n),G(n)}_{j}=\tilde{\alpha}_{j}^{p^2_1,C_1} \label{Eq4.16}	
\end{eqnarray}
\end{proposition}

\begin{proof}
For (\ref{Eq4.15}) we observe that
\begin{eqnarray}
&& \lim_{n\longrightarrow \infty}\epsilon_n \beta^{F(n),G(n)}_{j}=\lim_{n\longrightarrow \infty}\sqrt{\frac{p_1^2}{2F(n)}}\frac{G(n)}{j(j+1)}=\\ \nonumber
&&  \lim_{n\longrightarrow \infty}\sqrt{\frac{p_1^2}{(j_0^1)^2+(n+1)^2-1}}\frac{(sign{(C_1)} j^1_0(n+1))}{j(j+1)}= \sqrt{p_1^2}\frac{(sign{(C_1)} j^1_0)}{j(j+1)}\\ \nonumber
&& \underbrace{=}_{(\ref{eq4.11})}  \sqrt{p_1^2}\frac{(sign{(C_1)} \sqrt{\frac{C_1^2}{p_1^2}})}{j(j+1)}= \frac{C_1}{j(j+1)}=\tilde{\beta}_{j}^{p^2_1,C_1}
\end{eqnarray}
 (\ref{Eq4.16}) is  obtained similarly.
\end{proof}
\begin{theorem}
For any  $|^{F(n),G(n)}_{j,m_j} \rangle$, $|^{F(n),G(n)}_{\tilde{j},\tilde{m}_{\tilde{j}}} \rangle$ $\in V_{F(n),G(n)}$ and any \newline $X\in so(4)$
\begin{eqnarray}\label{eq4.17}
&&\lim_{n\rightarrow \infty}\left\langle^{F(n),G(n)}_{\tilde{j},\tilde{m}_{\tilde{j}}}|\rho_{F(n),G(n)}(t_n(X))  |^{F(n),G(n)}_{j,m_j} \right\rangle=\\ \nonumber
&& \left\langle^{p^2_1,C_1}_{\tilde{j},\tilde{m}_{\tilde{j}}}|\eta_{p^2_1,C_1}(\psi(X))  |^{p^2_1,C_1}_{j,m_j} \right\rangle
\end{eqnarray}
where $t_n=t(\epsilon_n)$.
\end{theorem}
\begin{proof}
We note that from linearity it is enough to prove that (\ref{eq4.17}) holds for $X\in \left\{M_+ ,M_- ,M_{3} ,N_+ ,N_- ,N_{3} \right\}$. We have:
\begin{eqnarray}
&& \lim_{n \longrightarrow \infty} \left\langle ^{F(n),G(n)}_{j,m_j}| \rho_{F(n),G(n)}(t_n (M_3))|^{F(n),G(n)}_{j,m_j} \right\rangle= \label{Eq.76}\\ \nonumber
&&\lim_{n \longrightarrow \infty} \left\langle ^{F(n),G(n)}_{j,m_j}| m_j|^{F(n),G(n)}_{j,m_j} \right\rangle=m_j=\left\langle^{p^2_1,C_1}_{j,m_{j}}|\eta_{p^2_1,C_1}(\psi(M_3))  |^{p^2_1,C_1}_{j,m_j} \right\rangle
\end{eqnarray}
\begin{eqnarray}
&& \lim_{n\longrightarrow \infty}\left\langle ^{F(j),G(j)}_{j,m_j\pm 1}| \rho_{F(n),G(n)}(t_n( M_{\pm}))|^{F(n),G(n)}_{j,m_j} \right\rangle=\\ \nonumber
&& \lim_{n\longrightarrow \infty} \sqrt{(j\mp m_j)(j\pm m_j+1)}= \sqrt{(j\mp m_j)(j\pm m_j+1)}= \\ \nonumber
&& =\left\langle^{p^2_1,C_1}_{j,m_j\pm 1}|\eta_{p^2_1,C_1}(\psi(M_{\pm}))  |^{p^2_1,C_1}_{j,m_j} \right\rangle
\end{eqnarray}
\begin{eqnarray}
&&\lim_{n\longrightarrow \infty} \sum_{k=-1}^{1} {\langle ^{F(n),G(n)}_{j+k,m_j}|  \rho_{F(n),G(n)} (t_{\epsilon_n}(N_{3}))|^{F(n),G(n)}_{j,m_j} \rangle} =\\ \nonumber
&& \lim_{n\longrightarrow \infty} \epsilon_n \alpha_{j}^{F(n),G(n)}\sqrt{(j+m_j)(j-m_j)}
 +\epsilon_n \beta^{F(n),G(n)}_j m_j|^{F(n),G(n)}_{j,m_j} +\\ \nonumber
&& \epsilon_n \alpha^{F(n),G(n)}_{j+1}\sqrt{(j+m_j+1)(j-m_j+1)}= \\ \nonumber
&&\tilde{\alpha}_{j}^{p_1^2,C_1}\sqrt{(j+m_j)(j-m_j)}
 +\tilde{\beta}^{p_1^2,C_1}_j m_j|^{F(n),G(n)}_{j,m_j} +\\ \nonumber
&& \tilde{\alpha}^{p_1^2,C_1}_{j+1}\sqrt{(j+m_j+1)(j-m_j+1)}= \sum_{k=-1}^{1} {\langle^{p^2_1,C_1}_{j+k,m_j}}|\eta_{p^2_1,C_1}(\psi(N_{3}))  |^{p^2_1,C_1}_{j,m_j} \rangle
\end{eqnarray}
where we have used proposition 1. Similarly:
\begin{eqnarray}
&& \lim_{n\longrightarrow \infty} \sum_{k=-1}^{1} {\langle ^{F(n),G(n)}_{j+k,m_j \pm 1}} |  \rho_{F(n),G(n)} (t_{\epsilon_n}(N_{\pm}))|^{F(n),G(n)}_{j,m_j} \rangle= \\ \nonumber
&&\lim_{n\longrightarrow \infty} \pm \epsilon_n \alpha^{F(n),G(n)}_{j}\sqrt{(j\mp m_j)(j\mp m_j-1)} +\\ \nonumber
&& \epsilon_n\beta^{F(n),G(n)}_j\sqrt{(j \mp m_j)(j\pm m_j+1)}\mp \\ \nonumber &&\epsilon_n\alpha^{F(n),G(n)}_{j+1}\sqrt{(j \pm m_j+1)(j \pm m_j+2)}= \pm \tilde{\alpha}^{p_1^2,C_1}_{j}\sqrt{(j\mp m_j)(j\mp m_j-1)} +\\ \nonumber
&& \tilde{\beta}^{p_1^2,C_1}_j\sqrt{(j \mp m_j)(j\pm m_j+1)}\mp \tilde{\alpha}^{p_1^2,C_1}_{j+1}\sqrt{(j \pm m_j+1)(j \pm m_j+2)}=\\ \nonumber
&& \sum_{k=-1}^{1} \langle^{p^2_1,C_1}_{j+k,m_j \pm 1}|\eta_{p^2_1,C_1}(\psi(N_{\pm}))  |^{p^2_1,C_1}_{j,m_j} \rangle
\end{eqnarray}
All the other matrix elements vanish and obviously satisfy (\ref{eq4.17}).
\end{proof}%\\[3Ex]

%\setcounter{footnote}{0}

%\newpage

\begin{figure}[h]
\scalebox{0.85}{\includegraphics{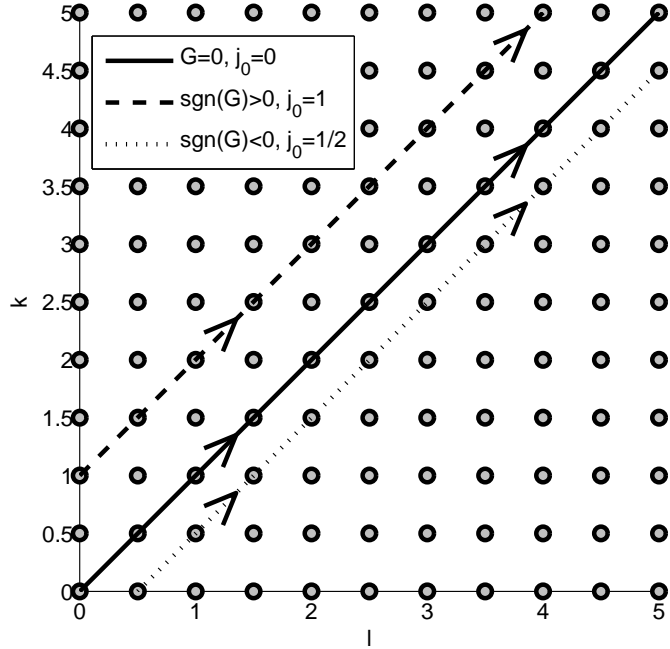}}
\caption{Grid of $so(4)$ finite dimensional irreducible representations and some of their contractions. Note the sign conventions
indicated in the box (upper left insert).}
\label{fig1}
\end{figure}
\begin{center}
\item \subsection{Graphical representation of the contraction process}
\end{center}
In figure \ref{fig1} each point with coordinates $(k,l)$ represents the irreducible representation of $so(4)$ which we denoted by $\rho^k\otimes \rho^l$. In each "diagonal" line, $j_0$ is constant and equal to the value of $|k-l|$ (those are the lines $k-l=\pm j_0$ in the $k,l$ plane). Going along each diagonal line in the direction of the arrow (which is equivalent to taking $\epsilon_n$ to zero) we are increasing the value of $n$ by one unit at each step , and this is the picture of the contraction. The solid, dashed and dotted diagonal lines correspond to contractions toward $\eta_{p^2,C}$ with their $j_0$ parameter equal to $0, 1$ and $\frac{1}{2}$ respectively. %\\[30Ex]

%\newpage

\section{Discussion}
The four-dimensional rotation group, $SO(4)$  occurs as a symmetry group of a physical system. The best known example is as the symmetry group of the Hydrogen atom. The group of isometries of the three-dimensional space, $\mathbb{R}^3$ i.e.,  the  Euclidean group $ISO(3)$ is another group that is  naturally related to many physical systems. Among others, $ISO(3)$ is a subgroup of both Poinca\'re group and Galilei group. The relation  between $SO(4)$  and $ISO(3)$ was only partially studied, e.g., \cite{Hol,Mon,Bob}.\newline
In another work \cite{ES,ES1} we give a  definition for contraction of Lie algebra representations using the notion of direct limit.  We also show there that the convergence of matrix elements implies the convergence in norm of the sequence of operators. This shows that the contraction we obtained here is also a contraction according to the definition in \cite{ES}.\newline

\textbf{Acknowledgments}\newline
AM is grateful to Prof. Weimar-Woods for a helpful discussion.\newline
EMS would like to thank Mr. ShengQuan Zhou for sharing his notes on the contractions of $so(4)$.   \newline
The research of the 2nd author was supported by
the center of excellence of the Israel Science Foundation
grant no. 1438/06.  \newline
JLB thanks the Department of Physics, Technion, for its warm
hospitality and support during visits while this work was
being carried out, and the FRAP-PSC-CUNY for some support.%\newline

%\section*{References}

\end{document}